\documentclass{article}
\usepackage[utf8]{inputenc}
\usepackage{amsmath,amssymb,bbm,amsthm,appendix,mathtools}
\usepackage[normalem]{ulem}
\usepackage[numbers,sort]{natbib}
\usepackage{graphicx}
\usepackage[scr=rsfs]{mathalpha}
\usepackage{cancel}
\usepackage[margin=1.2in]{geometry}
\usepackage[dvipsnames]{xcolor}
\usepackage{footmisc}
\usepackage{wrapfig}
\usepackage{comment}

\usepackage{color}

\newcommand\red[1]{{\color{red}#1}}
\newcommand{\Q}{\mathcal{Q}}
\newcommand{\loc}{\mathrm{loc}}
\newcommand{\Sp}{\mathbb{S}}
\newcommand{\id}{\mathrm{id}}

\newcommand{\Mw}{M_{\mathrm{wns}}}

\newcommand{\gw}{g_{\mathrm{wns}}}

\newcommand{\rd}{\partial}

\newcommand{\vol}{\text{vol}}
\newcommand{\R}{\mathbb{R}}

\usepackage{todonotes}

\usepackage{marginnote}

\newtheorem{definition}{Definition}[section]
\newtheorem{example}[definition]{Example}
\newtheorem{prop}[definition]{Proposition}
\newtheorem{remark}[definition]{Remark}
\newtheorem{lemma}[definition]{Lemma}
\newtheorem{thm}[definition]{Theorem}

\newtheorem{cor}[definition]{Corollary}
\usepackage{graphicx} 
\mathtoolsset{showonlyrefs,showmanualtags}

\title{Spherically symmetric inextendibility of weak null singularities with Christoffel symbols in $L^s_{\loc}$}

\author{\bigskip Peter Cameron\footnote{Department of Mathematics, Imperial College London, South Kensington Campus, London SW7 2AZ, United Kingdom and the Heilbronn Institute for
Mathematical Research, Bristol, UK. Email: p.cameron24@imperial.ac.uk} \hspace{1pt} \& Jan Sbierski\footnote{School of Mathematics and Maxwell Institute for Mathematical Science, University of Edinburgh, James Clerk Maxwell Building, Peter Guthrie Tait Road, Edinburgh EH9 3FD, United Kingdom. Email: jan.sbierski@ed.ac.uk}
}
\date{\today}

\begin{document}

\maketitle
\begin{abstract}
Motivated by the strong cosmic censorship conjecture in general relativity, we prove the inextendibility of spherically symmetric weak null singularities as spherically symmetric Lorentzian manifolds with a continuous metric and Christoffel symbols in $L^{s}_{\loc}$ for $s >1$. The result assumes a suitable blow-up condition on the derivative of the area-radius function transverse to the weak null singularity in combination with a rigidity result on continuous spherically symmetric extensions across null boundaries proven in \cite{CameronSbierski}. In particular we show that these assumptions are satisfied by the Reissner-Nordstr\"{o}m-Vaidya spacetime as well as by a class of spacetimes arising from small and generic spherically symmetric perturbations of subextremal Reissner-Nordstr\"{o}m under the Einstein-Maxwell-scalar field system.

\end{abstract}

\begingroup
\setlength{\parskip}{0pt}
\tableofcontents
\endgroup

\section{Introduction}\label{Introduction}
The modern formulation of the strong cosmic censorship conjecture, due to Christodoulou \cite{Chris09}, states that the maximal globally hyperbolic development of generic asymptotically flat or compact initial data for the (vacuum) Einstein equations is inextendible as a Lorentzian manifold with a continuous metric and locally square-integrable Christoffel symbols. An important case of interest for this conjecture arises at so-called weak null singularities which form in the interior of generic asymptotically flat vacuum black holes. It was shown in the seminal work of Dafermos-Luk \cite{DafLuk17, DafLuk26} that, for black holes which settle down to Kerr at a suitable rate along the event horizon, there exists a non-empty Cauchy horizon in the black hole interior across which the spacetime metric extends continuously. Furthermore, it was shown independently by Gurriaran \cite{Gurriaran.nonlinear} and Luk-S.\ \cite{LukSbie26} that, under suitable lower bounds on the decay of the perturbation to Kerr, the Cauchy horizon becomes weakly singular. Moreover, in \cite{SbieLip} the second author proved the Lipschitz-inextendibility of such weak null singularities. However, until now no geometric inextendibility results at the level of locally square-integrable Christoffel symbols have been available.
In this paper we prove such an inextendibility result (in fact we prove inextendibility with a continuous metric and Christoffel symbols in $L^s_{\loc}$ for $s>1$), although we restrict to spherically symmetric weak null singularities and  extensions which respect this spherical symmetry. 

The first example we apply our result to is the Reissner-Nordstr\"{o}m-Vaidya spacetime. This spacetime models a spherically symmetric influx of null dust into a subextremal Reissner-Nordstr\"{o}m black hole.\footnote{We assume an inverse polynomial decay of the accretion rate.} It was the earliest exact solution used to understand singularity formation at the Cauchy horizon in dynamically charged or rotating black holes, see \cite{Hiscock1981}. In \cite{Sbie22a} it was shown that this spacetime is $C^{0,1}_{\loc}$-inextendible across the weak null singularity. Here, we strengthen this to $C^0 \cap W^{1,s}_{\loc}$-inextendibility for extensions which respect the spherical symmetry. Note that the Hawking mass remains uniformly bounded at the weak null singularity of the Reissner-Nordstr\"om-Vaidya spacetime.

The second class of spacetimes we apply our result to arise as the maximal Cauchy development of generic perturbations of exact subextremal Reissner-Nordstr\"om initial data under the spherically symmetric Einstein-Maxwell-scalar field system (for the sake of brevity we call this class of spacetimes the `DLO spacetimes'). The study of these spacetimes goes back to Dafermos, who proved in \cite{Daf03,Daf05a} that, under the assumption of a pointwise upper bound on the decay of the scalar field along the event horizon, a non-empty Cauchy horizon forms in the black hole interior across which the spacetime is $C^0$-extendible. The assumed decay bound is consistent with Price's law \cite{Price72} and was later established by Dafermos and Rodnianski in \cite{DafRod05}. In the same work \cite{Daf05a} Dafermos showed that if additionally a pointwise lower bound on the decay of the scalar field holds along the event horizon, then the Cauchy horizon becomes weakly singular in the sense that the Hawking mass blows up there. This directly translates to the $C^{0,1}_{\loc}$-inextendibility of these weak null singularities in the spherically symmetric category of extensions (cf.\ Lemma \ref{lemma:SSSLip}). The pointwise lower bound assumed in this result was only proven recently by Gautam in \cite{Gautam}. 

An analysis of DLO spacetimes based on integrated lower bounds was carried out by Luk and Oh in \cite{LukOh19I,LukOh19II} (see also \cite{LukOhShla23}), who gave a complete proof that the maximal future Cauchy development arising from generic initial data is $C^2$-inextendible (not restricted to spherically symmetric extensions). This unrestricted $C^2$-inextendibility was strengthened in  \cite{Sbie22a} by the second author to unrestricted $C^{0,1}_{\loc}$-inextendibility.

Returning to the work of Dafermos \cite{Daf05a}, he explicitly constructed a \emph{particular} set of coordinates in which the metric extends continuously and in a spherically symmetric manner to the weak null singularity and in which the Christoffel symbols extend in $L^1_{\loc}$ but not in $L^{s}_{\loc}$ for any $s>1$. The result obtained in this paper strengthens  the spherically symmetric $C^{0,1}_{\loc}$-inextendibility from \cite{Daf05a} to spherically symmetric $C^0\cap W^{1,s}_{\loc}$-inextendibility for $s>1$ and thus shows that the regularity of the extension obtained in \cite{Daf05a} is optimal among extensions which respect spherical symmetry.

\subsection{The main theorem}

The manifolds we consider in this paper are given by $(\Mw,\gw)$, where $\Mw=\Q\times\Sp^2$; $\Q=(-1,1)\times(-1,0)$ with coordinates $(u, v)$, and $(\theta^1, \theta^2) = \theta^A$, $A=1,2$ are local smooth coordinates on $\Sp^2$. The metric in double null form is
\begin{equation} \label{EqGDN}
\begin{split}
    \gw&:=\underbrace{-\Omega^2(du \otimes dv + dv \otimes du)}_{=:g_\Q} + r^2\gamma_{AB}d\theta^A\otimes d\theta^B\;,
\end{split}
\end{equation}
where $\Omega(u,v)$ and $r(u,v)$ are smooth, strictly positive functions on $\Mw$ depending only on $(u,v)$ and $\gamma$ is the usual round metric on $\Sp^2$. The time orientation is fixed by stipulating that $\rd_u$ is future-directed null. The Lorentzian manifold $(\Mw, \gw)$ should be thought of as a \emph{local piece} of the 
\begin{wrapfigure}[17]{r}{0.4\textwidth}
 \centering
  \def\svgwidth{4cm}
   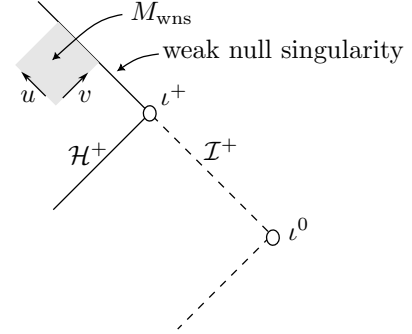 
   \caption{Penrose-style diagram of a spherically symmetric black hole with a weak null singularity.}
   \label{FigPen}
\end{wrapfigure}
spacetime near a weak null singularity in the interior of a spherically symmetric black hole (arising from the Einstein equations coupled to suitable matter fields, see also Section \ref{Application to  weak null singularities in the interior of spherically symmetric black holes} and Figure \ref{FigPen}). The weak null singularity occurs at $v=0$ and is signaled by $\rd_v r$ becoming singular as $v \nearrow 0$. 
It is called ``weak'' because $\gw$ extends continuously to $\{v=0\}$ with respect to the $(u,v,\theta^A)$ differentiable structure. This means that $\Omega$ and $r$ extend continuously to $\overline{\Mw}=\overline{Q}\times\Sp^2:=(-1,1)\times(-1,0]\times\Sp^2$ as strictly positive functions. \emph{From now on we will assume this holds for the class of Lorentzian manifolds considered in this paper.} The singularity is referred to as `null' because surfaces of constant $v$ are null hypersurfaces in $(\Mw,\gw)$ (this name is further justified by \cite[Proposition 3.2 (b)]{CameronSbierski}). 
As explained in detail in \cite[Section 2.1.1]{SbieLip} the inextendibility `across' the global weak null singularity in Figure \ref{FigPen} follows if each such local piece  $(\Mw, \gw)$ is inextendible `across' $v=0$ -- a notion which we recall now.

Given a regularity class $\Gamma$ (e.g. $\Gamma=C^0,C^{0,1}_{\loc},C^0\cap W^{1,s}_{\loc}$) we say that a $\Gamma$-extension of a Lorentzian manifold $(M,g)$ is an isometric embedding $\iota:M\hookrightarrow\tilde{M}$ of $(M,g)$ into a Lorentzian manifold $(\tilde{M},\tilde{g})$ of the same dimension as $M$, where $\tilde{g}$ is $\Gamma$-regular and $\partial\iota(M)\subset\tilde{M}$ is non-empty.
If a $\Gamma$-extension exists we say that $(M,g)$ is \textbf{$\Gamma$-extendible}, otherwise we say $(M,g)$ is \textbf{$\Gamma$-inextendible}. 
For the spacetime $(\Mw,\gw)$, a \textbf{$\Gamma$-extension of $(\Mw,\gw)$ across $\{v=0\}$} is a $\Gamma$-extension $\iota:\Mw\hookrightarrow\tilde{M}$ of $(\Mw,\gw)$ such that there exists a $C^1$ future-directed causal curve $\tau : [-1,0) \to \Mw$ with $\lim_{s \to 0} \tau_{v}(s) = 0$, $\lim_{s \to 0} \tau_u(s) <1$ and such that $\lim_{s \to 0}(\iota \circ \tau)(s) \in \rd \iota(\Mw) \subset \tilde{M}$ exists. 

In this paper we restrict attention to $\Gamma$-extensions of $(\Mw,\gw)$ (across $\{v=0\}$) of the form $\iota=\iota_{\Q}\times\text{id}_{\Sp^2}:\Mw\rightarrow\tilde{M}=\tilde{\Q}\times\Sp^2$, where $\iota_{\Q}:\Q\hookrightarrow\tilde{\Q}$ is a $\Gamma$-extension of $(\Q,g_{\Q})$ (across $\{v=0\}$) into $(\tilde{\Q}, \tilde{g}_{\tilde{\Q}})$, the metric $\tilde{g}$ on $\tilde{M}$ is of the form $\tilde{g} = \tilde{g}_{\tilde{\Q}} + \tilde{r}^2 \gamma$ with $\tilde{r}$ being a positive $\Gamma$-regular function on $\tilde{\Q}$ which extends $r \circ \iota_{\Q}^{-1}$, and  $\text{id}_{\Sp^2}$ denotes the identity map on the $\Sp^2$ factor. We refer to such extensions as \textbf{strongly spherically symmetric}.\footnote{The reason for calling such extensions `strongly' spherically symmetric is that even if  the original spacetime admits alternative foliations by orbits of the $SO(3)$-action (like for example Minkowski spacetime), the definition fixes the orbits of the $SO(3)$-action on the extension.} This restriction allows us to apply many of the results of \cite{CameronSbierski} to $\iota_\Q$.

As in \cite{CameronSbierski} (and in contrast to \cite{sbierski2022uniqueness}) we will refer to the fixed ``reference'' $C^0$-extension $\overline{\Q}$ for the purposes of discussing the uniqueness of $C^0$-extensions of $(\Q,g_\Q)$.\footnote{To be precise, $(\overline{\Q}, g_\Q)$ is not a $C^0$-extension of $(\Q, g_\Q)$ in the sense defined earlier, since $(\overline{\Q}, g_\Q)$ is a manifold \emph{with} boundary, whereas an extension is defined to be an isometric embedding into a manifold \emph{without} boundary. Hence it is more accurate to call $(\overline{\Q}, g_\Q)$ a \emph{continuous boundary extension of $(\Q, g_\Q)$.} However, it is straightforward to construct a proper reference $C^0$-extension $\iota_1 : \Q\hookrightarrow \tilde{\Q}_1$ from this.} In particular, suppose $\iota_\Q:Q\to\tilde{\Q}$ is a $C^0$-extension of $(\Q,g_\Q)$ such that $\tilde{p}_{u_*}:=\lim_{v \to 0}\iota_{\Q}(u_*,v) \in \partial\iota_\Q(\Q)$ exists for some $u_*\in(-1,1)$.\footnote{Remark \ref{RemEquivDefAcross} shows that this assumption is, in $1+1$ dimensions, equivalent to the assumption of $\iota_{\Q} : \Q \to \tilde{\Q}$ being a $C^0$-extension \emph{across $\{v=0\}$}. In this way we do not have to keep track of the arbitrary causal curve $\tau$.} If there exists $\epsilon>0$ such that $\iota_{\Q}\vert_{(u_* - \epsilon, u_* + \epsilon) \times (-\epsilon,0) } : (u_* - \epsilon, u_* + \epsilon) \times (-\epsilon,0)\to \tilde{\Q}$ extends  to $$\overline{\iota}_{\Q} : (u_* - \epsilon, u_* + \epsilon) \times (-\epsilon,0] \to \tilde{\Q}$$
as a $C^k$-diffeomorphism/homeomorphism onto its image then we say that $\iota_\Q : \Q \hookrightarrow \tilde{\Q}$ is \textbf{locally $C^k$/($C^0$)-equivalent to $\overline{\Q}$ at $(u_*, 0) \in \overline{\Q}$}.
If no such $\epsilon>0$ exists, we say that $\iota_\Q : \Q \hookrightarrow \tilde{\Q}$ is \textbf{locally $C^k$/($C^0$)-inequivalent to $\overline{\Q}$ at $(u_*, 0) \in \overline{\Q}$}. We can now state the main theorem of this paper.

\begin{thm}\label{thm:Christoffelblowup}
Let $(\Mw,\gw)$ be as above with $\Omega$ and $r$ extending as continuous positive functions to $\overline{\Q}$ and let $s>1$. Let $u_*\in(-1,1)$ and suppose that for all $\epsilon>0$ sufficiently small 
\begin{equation}\label{eqn:blowupassumption}
    \begin{split}
       \int_{-\frac{1}{2}}^v \inf_{u\in[u_*-\epsilon,u_*+\epsilon]}  \left|\partial_vr(u,v')\right|^{s} dv'\to\infty\text{ as }v\to0\;.
    \end{split}
\end{equation} 
Then there is no strongly spherically symmetric $C^0\cap W^{1,s}_{\loc}$-extension, $\iota:=\iota_\Q\times\id_{\Sp^2} :\Mw\rightarrow\tilde{M}=\tilde{\Q}\times\Sp^2 $, of $(\Mw,\gw)$ 
such that $\lim_{v\to0}\iota_\Q(u_*,v) \in \tilde{\Q}$ exists and $\iota_\Q$ is locally $C^0$-equivalent to $\overline{\Q}$ at $(u_*,0)$.
\end{thm}

 Note that Theorem \ref{thm:Christoffelblowup} only rules out $C^0 \cap W^{1,s}_{\loc}$-extensions where $\iota_\Q$ is locally $C^0$-equivalent to $\overline{\Q}$ at $(u_*,0)$. 
In Section \ref{The Corner Extension}, we show that this restriction is necessary in general: we give an example of an $(\Mw, \gw)$ satisfying the assumptions of Theorem \ref{thm:Christoffelblowup} for which there exists a strongly spherically symmetric $C^0 \cap W^{1,s}_{\loc}$-extension across $\{v=0\}$ such that $\iota_\Q$ is \emph{not} $C^0$-equivalent to $\overline{\Q}$ at $(u_*,0)$. Hence, when applying Theorem \ref{thm:Christoffelblowup} to the examples of the Reissner-Nordstr\"om-Vaidya spacetime in Section \ref{RNV} (see Corollary \ref{cor:RNV})  and the DLO spacetimes in Section \ref{DLO} (see Corollary \ref{cor:DLO}) to conclude their strongly spherically symmetric $C^0 \cap W^{1,s}_{\loc}$-inextendibility across the weak null singularity, we need to rule out the existence of $C^0$-extensions\footnote{Technically, ruling out the existence of such $C^0 \cap W^{1,s}_{\loc}$-extensions would suffice, but we will not use the additional information.} across the weak null singularity at $\{v=0\}$ which are locally $C^0$-\emph{inequivalent} to $\overline{\Q}$. In the case of the Reissner-Nordstr\"om-Vaidya spacetime, the obstruction to the existence of such $C^0$-extensions is entirely \emph{global}. For each \emph{local} piece $(\Mw, \gw)$ a so-called `corner' $C^0$-extension across $\{v=0\}$ can be constructed in which all outgoing null geodesics approach the same limit point. This construction crucially uses that the area radius function extends as a constant to the weak null singularity. Moreover, it is shown in \cite[Theorem 4.3]{CameronSbierski} that in this case the corner extension must be \textit{global}, meaning that \textit{all} outgoing null geodesics in the black hole interior approach the same limit point. However, as shown in \cite[Proposition 4.4]{CameronSbierski} (see also \cite[Corollary 4.6]{CameronSbierski}) such $C^0$-inequivalent extensions of the \emph{global} spacetime are ruled out by the infinite spacetime volume between the event horizon and the singular Cauchy horizon. The same obstruction is also present in the case of the DLO spacetimes. However, here even a \emph{local} obstruction is present which we appeal to in Section \ref{DLO}: the area radius function is generically strictly monotonically decreasing along the weak null singularity which even rules out the corner extensions of the local spacetime pieces $(\Mw, \gw)$.

Our result can also be applied to the spherically symmetric weak null singularities arising in the original mass-inflation spacetimes of Poisson-Israel \cite{Poisson:1989zz, Poisson:1990eh} (see also \cite{Ori:1991zz}) and the perturbations of subextremal Reissner-Nordstr\"om under the spherically symmetric Einstein-Maxwell-charged scalar field system studied by van de Moortel \cite{Vdm18, VdM21, VdM23}.

\subsection{Discussion of main theorem}

There has been very little previous work on strongly spherically symmetric extensions. An example of an inextendibility result with this restriction can be found in \cite{GalLin16}. This applies to a certain class of hyperbolic FLRW spacetimes. Another result for this class of extensions derives from Theorem \ref{thm:boundaryrigidity} (originally \cite[Theorem 4.3]{CameronSbierski}) which classifies the $C^0$-structure of $C^0$-extensions of $(\Q,g_\Q)$ across $\{v=0\}$ and forms the basis for ruling out the global `corner' extensions mentioned above. 

Clearly, the assumption of spherical symmetry on the extension significantly simplifies any inextendibility proof. For example, in the lemma below we show how the strongly spherically symmetric $C^{0,1}_{\loc}$-inextendibility of $(\Mw,\gw)$ can easily be inferred from mass inflation, where we recall the definition $m_{\mathrm{Hawking}} := \frac{r}{2}\big(1 - g_{\mathrm{wns}}^{-1}(dr,dr)\big)$ of the Hawking mass in spherical symmetry.

\begin{lemma}\label{lemma:SSSLip}
Let $(\Mw,\gw)$ be as in Section \ref{Introduction} and $u_*\in(-1,1)$. Suppose there exists $c>0$ such that
\begin{equation}\label{eqn:blowupassumptionLip}
    \begin{split}
      \partial_ur(u_*,v)&<-c<0\\
       \text{and }\quad        \partial_vr(u_*,v)&\to-\infty \quad \text{ as }v\to0.
    \end{split}
\end{equation}
Then there is no strongly spherically symmetric $C^{0,1}_{\loc}$-extension $\iota=\iota_{\Q}\times {\id}_{\Sp^2}:\Mw\to\tilde{M}$ of $(\Mw,\gw)$ such that $\lim_{v\to0}\iota_{\Q}(u_*,v)$ exists.
\end{lemma}
\begin{proof} 
    Suppose $\iota=\iota_{\Q}\times \text{id}_{\Sp^2}:\Mw\hookrightarrow\tilde{M}$ is a strongly spherically symmetric $C^0$-extension of $(\Mw,\gw)$ into some Lorentzian manifold $(\tilde{M},\tilde{g})$ such that $\lim_{v\to0}\iota_{\Q}(u_*,v)$ exists. For any $\theta_*^A\in\Sp^2$, we have
    \begin{equation}
        \begin{split}
           \left|\left.\tilde{g}^{-1}(d\tilde{r},d\tilde{r})\right\vert_{\iota(u_*,v,\theta^A_*)}\right|
           &=\left|\left.\gw^{-1}(dr,dr)\right\vert_{(u_*,v,\theta^A_*)}\right|\\
            &=2\Omega(u_*,v)^{-2}\left|\partial_ur(u_*,v)\partial_vr(u_*,v)\right|\\
            &\to\infty\text{ as }v\to0.
        \end{split}
    \end{equation}
It follows that $\iota$ is not a $C^{0,1}_{\loc}$-extension.
\end{proof}

Abstractly we have used here that in the spherically symmetric category, $g^{-1}(dr,dr)$ is a scalar function at the level of first derivatives of the metric. Hence, its blow-up is a very simple pointwise obstruction to strongly spherically symmetric $C^{0,1}_{\loc}$-extendibility in very much the same way as the blow-up of the Kretschmann scalar is a simple pointwise obstruction to (unrestricted) $C^{1,1}_{\loc}$-extendibility. Moving from here to strongly spherically symmetric $C^0 \cap W^{1,s}_{\loc}$-inextendibility poses two challenges: the first is that obstructions to $C^0 \cap W^{1,s}_{\loc}$-extendibility are manifestly local (integral) rather than pointwise: divergence must be detected on small neighbourhoods rather than 
along a sequence of points. Thus, the local topological ($C^0$) structure of such extensions has to be understood (this has been achieved in \cite{CameronSbierski}). The second challenge is that the $C^1$-differentiable structure of the extension has to be controlled: the scalar function $g^{-1}(dr,dr)$ cannot be used to obtain a \emph{sharp} inextendibility result of weak null singularities since it is only $\rd_v r$ which blows up, while $\rd_u r$ stays regular (and may even vanish). Hence, any use of $g^{-1}(dr,dr)$ would be intrinsically wasteful. This leaves us working with the singular Christoffel symbol (or derivative of a metric quantity) directly. Here, the assumption of spherical symmetry provides the simplification that $r$ is a scalar on the extension and hence $dr$, or equivalently the Christoffel symbols
\begin{equation}
    \Gamma^A_{Bu}=\frac{\partial_ur}{r} \delta^A_B\quad\text{and}\quad\Gamma^A_{Bv}=\frac{\partial_vr}{r} \delta^A_B\,,
\end{equation}
transform like tensor quantities -- in contrast to their transformation rule involving the $C^2$-differentiable structure in the absence of spherical symmetry. 

The control of the $C^1$-differentiable structure needed for the proof of Theorem \ref{thm:Christoffelblowup} (which is by contradiction) is obtained in two parts. The first is contained in Lemma \ref{lemma:derivativecontrol}, which relies on the fact that although derivatives of $\iota_\Q(u,v)$ may blow-up as $v\to0$, the continuity of the metric allows us to bound certain derivatives in terms of others. The second part is obtained in \eqref{EqContrC1} through the assumption that the $C^0\cap W^{1,s}_{\loc}$-extension is locally $C^0$-equivalent to the reference extension. 

Finally, we point out that in the proof \cite{SbieLip} of Lipschitz-inextendibility of weak null singularities without any symmetry assumptions it was also a crucial step to a priori control the $C^1$-differentiable structure of any assumed $C^{0,1}_{\loc}$-extension. This was achieved in \cite{sbierski2022uniqueness}, where it was shown that any such $C^{0,1}_{\loc}$-extension is $C^1$-equivalent to the reference extension. On the other hand, in \cite{CameronSbierski} it was shown explicitly that such strong control on the $C^1$-differentiable structure  need not hold for $C^0 \cap W^{1,s}_{\loc}$-extensions. Here, we establish weaker control bounds on the $C^1$-differentiable structure which are, however, still strong enough to yield a $C^0 \cap W^{1,s}_{\loc}$-inextendibility statement.

\subsection*{Acknowledgements}

The first author was supported by the Additional Funding Programme for
Mathematical Sciences, delivered by EPSRC (EP/V521917/1) and the
Heilbronn Institute for Mathematical Research. The second author was supported by the Royal Society University Research Fellowship URF\textbackslash R1\textbackslash 211216. 

\section{Background results}

In this section we collect some background results on $C^0$-extensions which will be required throughout the paper. 

\begin{lemma}[\cite{SbierskiSchw} Lemma 2.4]\label{lemma:nearmink}
 Let $(M,g)$ be a $(1+1)$-dimensional Lorentzian manifold with continuous metric $g$ and let $p\in M$. Then for every $\delta>0$ we can find $\epsilon_0,\epsilon_1>0$, an open neighbourhood $U$ of $p$, and a coordinate chart $\varphi:U\rightarrow(-\epsilon_0,\epsilon_0)\times(-\epsilon_1,\epsilon_1)$ such that
    \begin{enumerate}
        \item $\varphi(p)=(0,0)$
        \item $g_{\mu\nu}(p)=m_{\mu\nu}$
        \item $|g_{\mu\nu}(p')-m_{\mu\nu}|<\delta$\text{ for all } $p'\in U$
    \end{enumerate}
where $m_{\mu\nu}=\begin{pmatrix}
    -1 & 0\\0 &1
\end{pmatrix}$ is the Minkowski metric on $\R^{1,1}$. It will be convenient to assume that $\delta>0$ is sufficiently small that
\begin{equation}
   \frac{1}{2}\leq\sqrt{-\det g}\leq 2\label{eqn:detbound}
\end{equation}
in $U$.
\end{lemma}

We refer to $(U,\varphi)$ as a \textbf{near Minkowski neighbourhood} centred at $p$ and to $\varphi$ as \textbf{near Minkowski coordinates}.

We will require the following basic structural result, proved in \cite{CameronSbierski}, for $C^0$-extensions in 1+1-dimensions.

\begin{prop}[\cite{CameronSbierski} Proposition 3.2 (a)(i)]\label{prop:boundarystructure}
Suppose $\iota_\Q:\Q\hookrightarrow\tilde{\Q}$ is a $C^0$-extension of $(Q,g_\Q)$ to $(\tilde{\Q},\tilde{g}_\Q)$. Suppose $\tau : [-1,0) \to \Q$ is a $C^1$ future-directed causal curve such that  $\lim_{s \to 0} \tau_u(s)=: u_* <1$, $\lim_{s \to 0} \tau_v(s) = 0$ and $\tilde{p}_{u_*}:=\lim_{s \to 0}(\iota_\Q \circ \tau)(s) \in \rd \iota_\Q(\Q) \subset \tilde{\Q}$ exists. 

Then there exists a near Minkowski neighbourhood, $(\tilde{U},\tilde{\varphi}=(\tilde{x}_0,\tilde{x}_1))$ centred at $\tilde{p}_{u_*}$ and $v_*<0<\epsilon$ such that $\iota_{\Q}:[u_*-\epsilon,u_*+\epsilon]\times[v_*,0)\hookrightarrow \tilde{U}$ extends continuously to $\iota_{\Q}:[u_*-\epsilon,u_*+\epsilon]\times[v_*,0]\hookrightarrow\tilde{U}$ and \begin{equation}\label{eqn:monotonicityconventions}
    \frac{\partial\tilde{x}_1}{\partial u}<0<\frac{\partial\tilde{x}_0}{\partial u},\frac{\partial\tilde{x}_0}{\partial v},\frac{\partial\tilde{x}_1}{\partial v}
\end{equation}
holds on the connected component of $\iota_{\Q}(\Q) \cap \tilde{U}$ that contains $\iota_{\Q}\big([u_* - \epsilon, u_* + \epsilon] \times [v_*,0)\big)$.\footnote{Here, and in the following, by slight abuse of notation we have denoted the restriction of $\iota_{\Q}$ to $[u_*-\epsilon,u_*+\epsilon]\times[v_*,0)$ as well its continuous extension to $[u_*-\epsilon,u_*+\epsilon]\times[v_*,0]$  by the same symbol.}

\end{prop}

\begin{remark} \label{RemEquivDefAcross}
    Note that the assumption of Proposition \ref{prop:boundarystructure} is exactly that $\iota_{\Q} : \Q \to \tilde{\Q}$ is a $C^0$-extension of $(\Q,g_{\Q})$ across $\{v=0\}$. The conclusion in particular implies that $\lim_{v \to 0} \iota_{\Q}(u_*,v) = \tilde{p}_{u_*}$, which shows that the $C^1$ future directed causal curve $\tau$ in the definition of `a $C^0$-extension of $(\Q,g_{\Q})$ across $\{v=0\}$' can, without loss of generality, be assumed to be a null curve of constant $u$.
\end{remark}

In Section \ref{Application to  weak null singularities in the interior of spherically symmetric black holes} we will use Theorem \ref{thm:Christoffelblowup} to prove inextendibility statements at the level of the Christoffel symbols across weak null singularities in the interiors of spherically symmetric black holes. Since Theorem \ref{thm:Christoffelblowup} only applies to extensions where $\iota_\Q$ is $C^0$-equivalent to $\overline{\Q}$, we first need to rule out the situation where this is not the case. To do this, we recall the following rigidity result. 
\begin{thm}[\cite{CameronSbierski} Theorem 4.3]\label{thm:boundaryrigidity}
  Let $\iota_\Q : \Q \hookrightarrow \tilde{\Q}$ be a $C^0$-extension of $(\Q,g_\Q)$ such that $\tilde{p}_{u_*}:=\lim_{v \to 0} \iota_\Q (u_*,v) \in \partial \iota_{\Q}(M)$ exists for some $u_*\in(-1,1)$. Then one of the following holds:
    \begin{enumerate}
        \item[(a)] $\iota_\Q$ extends continuously to $\iota_\Q : (-1,1) \times (-1,0] \to 
        \tilde{\Q}$ with $\iota_\Q(u,0) = \tilde{p}_{u_*}$ for all $u \in (-1,1)$. 
        \item[(b)] 
        There exist $v_*<0<\epsilon$
        such that $\iota_\Q$ extends to $\iota_\Q : [u_*-\epsilon, u_*+\epsilon] \times [v_*,0] \to \tilde{\Q}$ as a homeomorphism onto its image.
    \end{enumerate}
\end{thm}


Given a strongly spherically symmetric $C^0$-extension of the Reissner-Nordstr\"{o}m-Vaidya spacetime considered in Section \ref{RNV}, in order to rule out the possibility that $\iota_\Q$ is a corner (by this we mean we are in case (a) in Theorem \ref{thm:boundaryrigidity}) we recall \cite[Proposition 4.4]{CameronSbierski}. 
We re-state this proposition so it applies to the 1+1-dimensional Lorentzian manifold $\Q_{RNV}=(u_{\mathcal{H}^+},u_T)\times(V_0,0)$ with coordinates $(u,V)$ and metric 
\begin{equation}
    g_{\Q_{RNV}}=-\Omega'(u,V)^2\left(du\otimes dV+dV\otimes du\right),
\end{equation}
where $\Omega'$ is a smooth, strictly positive function on $\Q_{RNV}$ which extends to $\overline{\Q_{RNV}}:=(u_{\mathcal{H}^+},u_T)\times(V_0,0]$ as a continuous, strictly positive function.
\begin{prop}[\cite{CameronSbierski} Proposition 4.4]\label{prop:cornervolumecondition}
   Let $\iota_{\Q_{RNV}}$ be a $C^0$-extension of $(\Q_{RNV},g_{\Q_{RNV}})$ to $(\tilde{\Q}_{RNV},\tilde{g}_{\tilde{\Q}_{RNV}})$ such that $\tilde{p}_{u_*}:=\lim_{V\rightarrow0}\iota_{\Q_{RNV}}(u_*,V)\in\tilde{\Q}_{RNV}$ exists for some $u_*\in(u_{\mathcal{H}^+},u_T)$. Suppose $\mathrm{vol}_{g_{\Q_{RNV}}}((u_{\mathcal{H}^+},u_T)\times(V_*,0)):=\int_{u_{\mathcal{H}^+}}^{u_T}\int_{V_*}^0\Omega'(u,V)^2dVdu=\infty$ for all $V_*\in(V_0,0)$. Then there exists $\epsilon>0$ such that $\iota_{\Q_{RNV}}$ extends to $\iota_{\Q_{RNV}}:[u_*-\epsilon,u_*+\epsilon]\times(V_0,0]\rightarrow\tilde{\Q}_{RNV}$ as a homeomorphism onto its image. In particular, $\iota_{\Q_{RNV}}$ is locally $C^0$-equivalent to $\overline{\Q}_{RNV}$ at $(u_*,0) \in \overline{\Q}_{RNV}$. 
\end{prop}

\section{The inextendibility result}
\subsection{Controlling the differentiable structure}\label{Controlling the differentiable structure}

Before proving Theorem \ref{thm:Christoffelblowup}, we revisit the setting of \cite{CameronSbierski} and consider $C^0$-extensions $\iota_{\Q} : \Q \to \tilde{\Q}$ of $(\Q, g_{\Q})$ across $\{v=0\}$. We will derive some control on the $C^1$-differentiable structure in the general setting of Proposition \ref{prop:boundarystructure}. In particular we show that although individual derivatives of $\iota_\Q(u,v)$ may blow-up, the continuity of $\tilde{g}$ allows us to bound certain derivatives of $\iota_\Q(u,v)$ (and its inverse) in terms of others.

\begin{lemma}\label{lemma:derivativecontrol} Consider the setting of Proposition \ref{prop:boundarystructure} with $(\tilde{U},\tilde{\varphi}=(\tilde{x}_0,\tilde{x}_1))$ being a near Minkowski neighbourhood centred at $\tilde{p}_{u_*}\in \partial\iota_\Q(\Q)$.  Define $p=\frac{1}{\sqrt{2}}(\tilde{x}_0-\tilde{x}_1)$ and $q=\frac{1}{\sqrt{2}}(\tilde{x}_0+\tilde{x}_1)$. Then for any\footnote{The estimate is of interest for small $\alpha >0$.} $\alpha>0$, after possibly reducing $\tilde{U}$, $|v_*|$, and $\epsilon$, we have
               \begin{equation}\label{eqn:derivbounds}
                    \begin{split}
                     \frac{1}{ 2}\left|\frac{\partial u}{\partial p}\right|&\leq \frac{1}{\Omega^2}\left|\frac{\partial q}{\partial v}\right|\leq  2\left|\frac{\partial u}{\partial p}\right|,\quad
        \frac{1}{ 2}\left|\frac{\partial u}{\partial q}\right|\leq \frac{1}{\Omega^2}\left|\frac{\partial p}{\partial v}\right|\leq  2\left|\frac{\partial u}{\partial q}\right|,\\
        \frac{1}{ 2}\left|\frac{\partial v}{\partial p}\right|&\leq \frac{1}{\Omega^2}\left|\frac{\partial q}{\partial u}\right|\leq  2\left|\frac{\partial v}{\partial p}\right|,\quad
        \frac{1}{ 2}\left|\frac{\partial v}{\partial q}\right|\leq \frac{1}{\Omega^2}\left|\frac{\partial p}{\partial u}\right|\leq  2\left|\frac{\partial v}{\partial q}\right|.
\end{split}
\end{equation}
and     
\begin{equation}\label{eqn:gradientboundderivs}
    \begin{split}
                \left|\frac{\partial u}{\partial p}\right|&\geq \frac{1}{\alpha}\left|\frac{\partial u}{\partial q}\right|,\quad\left|\frac{\partial v}{\partial q}\right|\geq \frac{1}{\alpha}\left|\frac{\partial v}{\partial p}\right|
                \end{split}
                \end{equation}
in $\iota_\Q\left([u_*-\epsilon,u_*+\epsilon]\times[v_*,0)\right) \subseteq \tilde{U}$.
\end{lemma}

Here, and in the following, we have omitted the obvious identifications, e.g. $\frac{\rd p}{\rd u}$ means $\frac{\rd (p\circ\iota_\Q)}{\rd u}$.

\begin{proof}We have 
\begin{equation}\label{eqn:derivativeinversion}
    \begin{split}
        \begin{pmatrix}
\frac{\partial u}{\partial p} & \frac{\partial u}{\partial q} \\
\frac{\partial v}{\partial p} & \frac{\partial v}{\partial q}
\end{pmatrix}
&=\frac{1}{\det J(u,v)}\begin{pmatrix}
\frac{\partial q}{\partial v} & -\frac{\partial p}{\partial v} \\
-\frac{\partial q}{\partial u} & \frac{\partial p}{\partial u}
\end{pmatrix}
    \end{split}
\end{equation}
in $\iota_\Q\left([u_*-\epsilon,u_*+\epsilon]\times[v_*,0)\right)$, where the determinant of the Jacobian matrix $$J(u,v)=\begin{pmatrix}
\frac{\partial p}{\partial u} & \frac{\partial p}{\partial v} \\
\frac{\partial q}{\partial u} & \frac{\partial q}{\partial v}
\end{pmatrix} = D \iota_{\Q}$$ can be estimated using \eqref{eqn:detbound} as follows:
\begin{equation}
    \begin{split}
        \left|\det J(u,v)\right|&=\frac{\sqrt{-\det g_\Q}}{\sqrt{-\det\tilde{g}_\Q}}\\
        &\in\left[\frac{1}{2}\Omega(u,v)^2,2\Omega(u,v)^2\right].
    \end{split}
\end{equation}
Combining this with \eqref{eqn:derivativeinversion}, we deduce that \eqref{eqn:derivbounds} holds in $\iota_\Q\left([u_*-\epsilon,u_*+\epsilon]\times[v_*,0)\right)$.

To prove \eqref{eqn:gradientboundderivs}, note that the metric at $\tilde{p}_{u_*}$ is
$$\tilde{g}_{\Q}\vert_{\tilde{p}_{u_*}}=-dp\otimes dq-dq\otimes dp$$
and hence $\partial_p$ and $\partial_q$ are null and linearly independent at $\tilde{p}_{u_*}$. 
We have 
\begin{equation}
      \frac{\partial\tilde{x}_1}{\partial p}<0<\frac{\partial\tilde{x}_0}{\partial p},\frac{\partial\tilde{x}_0}{\partial q},\frac{\partial\tilde{x}_1}{\partial q}
\end{equation} 
and hence comparing with the monotonicity properties \eqref{eqn:monotonicityconventions} we see that$\frac{(\iota_\Q)_*\partial_u}{\lVert(\iota_\Q)_*\partial_u\rVert_{\R^2}}$ and $\frac{(\iota_\Q)_*\partial_v}{\lVert(\iota_\Q)_*\partial_v\rVert_{\R^2}}$ defined in $\iota_\Q\left([u_*-\epsilon,u_*+\epsilon]\times[v_*,0)\right)$ extend continuously to $\partial_p$ and $\partial_q$ respectively at $\tilde{p}_{u_*}$.  Here $\lVert.\rVert_{\R^2}$ denotes the norm defined with respect to the Riemannian metric $\delta=d\tilde{x}_0\otimes d\tilde{x}_0+d\tilde{x}_1\otimes d\tilde{x}_1$ on $\tilde{U}$ and 
\begin{equation}
\begin{split}
        (\iota_\Q)_*\partial_u&=\frac{\partial p}{\partial u}\partial_p+\frac{\partial q}{\partial u}\partial_q\\
     (\iota_\Q)_*\partial_v&=\frac{\partial p}{\partial v}\partial_p+\frac{\partial q}{\partial v}\partial_q.
\end{split}
\end{equation}
Let $L$ and $\underline{L}$ be two continuous, linearly independent null vector fields in $\tilde{U}$ such that $L\vert_{\tilde{p}_{u_*}}=\partial_p$ and $\underline{L}\vert_{\tilde{p}_{u_*}}=\partial_q$. We see that, in $\iota_\Q\left([u_*-\epsilon,u_*+\epsilon]\times[v_*,0)\right)$, $L$ and $\underline{L}$ are proportional to $(\iota_\Q)_*\partial_u$ and $(\iota_\Q)_*\partial_v$ respectively.

Let $\alpha>0$. By continuity, we can choose $\tilde{U}$ sufficiently small that 
\begin{equation}
    \frac{\langle\partial_p,L\rangle_{\R^2}}{\lVert L\rVert_{\R^2}},\frac{\langle\partial_q,\underline{L}\rangle_{\R^2}}{\lVert \underline{L}\rVert_{\R^2}}\geq\frac{1}{\sqrt{1+\left(\frac{\alpha}{4}\right)^2}}
\end{equation}
in $\tilde{U}$, where $\langle.,.\rangle_{\R^2}$ denotes the inner product defined with respect to $\delta$.
It follows that
\begin{equation}\label{eqn:gradientcontrol}
    \begin{split}
        \left|\frac{\partial q}{\partial u}\right|&\leq\frac{\alpha}{4} \left|\frac{\partial p}{\partial u}\right|\\
        \text{and }\left|\frac{\partial p}{\partial v}\right|&\leq\frac{\alpha}{4}\left|\frac{\partial q}{\partial v}\right|
    \end{split}
\end{equation}
in $\iota_\Q\left([u_*-\epsilon,u_*+\epsilon]\times[v_*,0)\right)$.


Combining \eqref{eqn:derivbounds} and \eqref{eqn:gradientcontrol}, it follows that in $\iota_\Q\left([u_*-\epsilon,u_*+\epsilon]\times[v_*,0)\right)$ we have
\begin{equation}
    \begin{split}
                \left|\frac{\partial u}{\partial p}\right|&\geq\frac{1}{ 2\Omega^2}\left|\frac{\partial q}{\partial v}\right|\\
                &\geq\frac{2}{\alpha\Omega^2}\left|\frac{\partial p}{\partial v}\right|\\
                &\geq  \frac{1}{\alpha}\left|\frac{\partial u}{\partial q}\right| \\
        \text{and similarly}\left|\frac{\partial v}{\partial q}\right|&\geq \frac{1}{\alpha}\left|\frac{\partial v}{\partial p}\right|.
    \end{split}
\end{equation}
\end{proof}

The proof of Theorem \ref{thm:Christoffelblowup} also requires the following elementary lemma.

\begin{lemma}\label{lemma:integralbalance}
    Let $\mathcal{D} \subseteq \R^2$ be open and let $f,h : \mathcal{D} \to \R$ be measurable functions. For each $v \in (-1,0)$ let $\mathcal{D}_v \subseteq \mathcal{D}$ be measurable. Let ${s}>1$ and assume $f \in L^{s}(\mathcal{D})$, $h \cdot \mathbbm{1}_{\mathcal{D}_v} \in L^{s}(\mathcal{D})$ for all $-1<v <0$, and $||h \cdot \mathbbm{1}_{\mathcal{D}_v}||_{L^{s}(\mathcal{D})} \to \infty$ as $v \to 0$. Then
    \begin{equation*}
        \frac{|| (f-h)\cdot \mathbbm{1}_{\mathcal{D}_v}||_{L^{s}(\mathcal{D})}}{||h \cdot \mathbbm{1}_{\mathcal{D}_v}||_{L^{s}(\mathcal{D})}} \to 1 \qquad \textnormal{ for } v \to 0 \;.
    \end{equation*}
\end{lemma}

\begin{proof}
    By the Minkowski inequality we have
    \begin{equation*}
        ||h \cdot \mathbbm{1}_{\mathcal{D}_v}||_{L^{s}(\mathcal{D})} - ||f||_{L^{s}(\mathcal{D})} \leq || (f-h)\cdot \mathbbm{1}_{\mathcal{D}_v}||_{L^{s}(\mathcal{D})} \leq ||h \cdot \mathbbm{1}_{\mathcal{D}_v}||_{L^{s}(\mathcal{D})} + ||f||_{L^{s}(\mathcal{D})} \;.
    \end{equation*}
    Division by $||h \cdot \mathbbm{1}_{\mathcal{D}_v}||_{L^{s}(\mathcal{D})}$ yields the result.
\end{proof}

\subsection{Proof of Theorem \ref{thm:Christoffelblowup}}

Using the results of Section \ref{Controlling the differentiable structure} we are now able to prove Theorem \ref{thm:Christoffelblowup}. Recall that the metric on $\iota(\Mw)$ is
\begin{equation}\label{eqn:extendedmetric}
    \tilde{g}=\tilde{g}_{\Q}+\tilde{r}^2\gamma_{AB}d\theta^A\otimes d\theta^B
\end{equation} 
where $\tilde{r}(\iota_{\Q}(u,v))=r(u,v)$ and $\tilde{g}_{\Q}=(\iota_{\Q})_*g_{\Q}$.


\begin{proof}[Proof of Theorem \ref{thm:Christoffelblowup}] 
Suppose for a contradiction that $\iota=\iota_\Q\times\id_{\Sp^2}:\Mw\to\tilde{M}=\tilde{\Q}\times\Sp^2$ is a strongly spherically symmetric $C^0 \cap W^{1,s}_{\loc}$-extension of $(\Mw,\gw)$ 
such that $\tilde{p}_{u_*}:=\lim_{v\to0}\iota_\Q(u_*,v)$ exists and $\iota_\Q$ is $C^0$-equivalent to $\overline{\Q}$ at $(u_*,0)$.
Let $(\tilde{u},\tilde{v})$ be local coordinates of $\tilde{\Q}$ in a neighbourhood of $\tilde{p}_{u_*}$ so that $\tilde{g} \in C^0 \cap W^{1,s}_{\loc}$ with respect to the $(\tilde{u}, \tilde{v}, \theta^A)$ coordinates on $\tilde{M}$.
Let $(\tilde{U},\tilde{\varphi}=(\tilde{x}_0,\tilde{x}_1))$ be a near Minkowski neighbourhood on $\tilde{\Q}$ centred at $\tilde{p}_{u_*}$ where, by the proof of \cite[Lemma 2.4]{SbierskiSchw}, $(\tilde{x}_0,\tilde{x}_1)$ are linearly related to $(\tilde{u},\tilde{v})$. We define $(p,q)$ by $p=\frac{1}{\sqrt{2}}(\tilde{x}_0-\tilde{x}_1)$ and $  q=\frac{1}{\sqrt{2}}(\tilde{x}_0+\tilde{x}_1)$ so that, after restricting $\tilde{U}$ if necessary, the 
conclusions of Lemma \ref{lemma:derivativecontrol} hold for some $\alpha\in\left(0,\frac{1}{2}\right)$. Note that the linear coordinate change from $(\tilde{u},\tilde{v},\theta^A)$ to $(p,q,\theta^A)$ does not affect whether or not 
$\tilde{g}$ lies in $C^0 \cap W^{1,s}_{\loc}$: the two charts are part of the same smooth structure on $\tilde{M}$. Hence, after restricting $\tilde{U}$ further if necessary, we conclude that $\partial_p\tilde{r}$ and $\partial_q\tilde{r}$ lie in $L^s(\tilde{U})$.

We have
\begin{equation} \label{EqDr}
   \begin{aligned}
      \frac{\partial \tilde{r}}{\partial p}&=\frac{\partial r}{\partial u}\frac{\partial u}{\partial p}+\frac{\partial r}{\partial v}\frac{\partial v}{\partial p}\\
     \frac{\partial \tilde{r}}{\partial q}&=\frac{\partial r}{\partial u}\frac{\partial u}{\partial q}+\frac{\partial r}{\partial v}\frac{\partial v}{\partial q}.
\end{aligned}
\end{equation}

In the remainder of the proof, we will contradict Lemma \ref{lemma:integralbalance} (with $f=\frac{\partial \tilde{r}}{\partial p}$ and $h=\frac{\partial r}{\partial u}\frac{\partial u}{\partial p}$) by showing that
\begin{equation} \label{EqBlowUp}
    \left\| \frac{\partial r}{\partial u}\frac{\partial u}{\partial p}\cdot \mathbbm{1}_{\iota_{\Q}([u_* - \epsilon, u_* + \epsilon]\times [v_*,v))}\right\|_{L^{s}(\tilde{U})}\to\infty \qquad \textnormal{ as } v\to0\;.
\end{equation}
but 
\begin{equation}\label{eqn:limitnot1}
        \frac{|| \frac{\partial r}{\partial v}\frac{\partial v}{\partial p}\cdot \mathbbm{1}_{\iota_{\Q}([u_* - \epsilon, u_* + \epsilon]\times [v_*,v))}||_{L^{s}(\tilde{U})}}{||\frac{\partial r}{\partial u}\frac{\partial u}{\partial p} \cdot \mathbbm{1}_{\iota_{\Q}([u_* - \epsilon, u_* + \epsilon]\times [v_*,v))}||_{L^{s}(\tilde{U})}} \not\to 1 \qquad \textnormal{ as } v \to 0 \;.
\end{equation}
The argument has three main steps. First, the blow-up assumption \eqref{eqn:blowupassumption} in combination with the local $C^0$-equivalence of the extension forces the $\partial_vr\partial_qv$ term in $\partial_q\tilde{r}$ to have divergent $L^s$-norm. Second, since $\partial_q\tilde{r}\in L^s$, the term $\partial_u r\partial_qu$ must asymptotically cancel it and therefore also has divergent $L^s$-norm. Lemma \ref{lemma:derivativecontrol} transfers this divergence to $\partial_ur\partial_pu$ thus showing \eqref{EqBlowUp}. Finally, in the third step, we use the same lemma as well as previous estimates on the relative sizes of terms in \eqref{EqDr} to show \eqref{eqn:limitnot1}.

Recall that $\Omega^2$ extends continuously to a strictly positive function $\Omega^2:[u_*-\epsilon,u_*+\epsilon]\times[v_*,0]\to\R_{>0}$. Hence there exists $C>0$ such that
\begin{equation}\label{eqn:Omegabound}
    \frac{1}{C}\leq \Omega^2(u,v)\leq C
\end{equation}
for all $(u,v)\in[u_*-\epsilon,u_*+\epsilon]\times[v_*,0]$.

It follows from \eqref{eqn:monotonicityconventions} that $\frac{\partial p}{\partial u}>0$ and hence $p(u_*+\epsilon,v)-p(u_*-\epsilon,v)>0$ for all $v\in[v_*,0)$. Since we are in case (b) of Theorem \ref{thm:boundaryrigidity}, we have $\lim_{v\rightarrow0}\iota_\Q(u_*-\epsilon,v)\neq \lim_{v\rightarrow0}\iota_\Q(u_*+\epsilon,v)$ and thus $p(u_*+\epsilon,0)-p(u_*-\epsilon,0)>0$ by \cite[Proposition 3.2 (b)]{CameronSbierski}. Hence there exists $C'>0$ such that 
\begin{equation}
    p(u_*+\epsilon,v)-p(u_*-\epsilon,v)\geq C'
\end{equation}
for all $v\in[v_*,0]$.

It follows from Lemma \ref{lemma:derivativecontrol} that
\begin{equation} \label{EqContrC1}
    \begin{split}
        \inf_{v'\in[v_*,0)}\int_{u_*-\epsilon}^{u_*+\epsilon} \left|\frac{\partial v}{\partial q}\right|^{s} du\geq& 2^{-s}\inf_{v'\in[v_*,0)}\int_{u_*-\epsilon}^{u_*+\epsilon}\Omega^{-2s} \left|\frac{\partial p}{\partial u}\right|^{s} du\\
        \geq& \frac{(2\epsilon)^{1-s}}{(2C)^{s} }\inf_{v'\in[v_*,0)}\left|\int_{u_*-\epsilon}^{u_*+\epsilon} \frac{\partial p}{\partial u}du\right|^{s}\\
        \geq& \frac{(2\epsilon)^{1-s}}{(2C)^{s} }\inf_{v'\in[v_*,0)}\left|(p(u_*+\epsilon,v')-p(u_*-\epsilon,v')\right|^{s}\\
        \geq& \frac{(2\epsilon)^{1-s}}{(2C)^{s} }C'^{s}
    \end{split}
\end{equation}
where in the second line we have used H\"{o}lder's inequality along with \eqref{eqn:Omegabound}.

Using \eqref{eqn:detbound}, the fact that $\iota_\Q$ is an isometry, and \eqref{eqn:blowupassumption}, we have
\begin{align}
        \int_{\iota_{\Q}([u_*-\epsilon,u_*+\epsilon]\times[v_*,v])}\left|\frac{\partial r}{\partial v}\frac{\partial v}{\partial q}\right|^{s} dpdq&\geq \frac{1}{2}\int_{\iota_{\Q}([u_*-\epsilon,u_*+\epsilon]\times[v_*,v])}\left|\frac{\partial r}{\partial v}\frac{\partial v}{\partial q}\right|^{s} \sqrt{-\det \tilde{g}_{\Q}}dpdq\\
        &= \frac{1}{2}\int_{[u_*-\epsilon,u_*+\epsilon]\times[v_*,v]}\left|\frac{\partial r}{\partial v}\frac{\partial v}{\partial q}\right|^{s} \Omega^2(u',v')du'dv'\\
        &\geq \frac{1}{2C}\int_{v_*}^v \inf_{u\in[u_*-\epsilon,u_*+\epsilon]}\left|\frac{\partial r}{\partial v}\right|^{s} dv' \inf_{v'\in[v_*,v)}\int_{u_*-\epsilon}^{u_*+\epsilon} \left|\frac{\partial v}{\partial q}\right|^{s} du'\\
        &\geq \frac{(2\epsilon)^{1-s}}{(2C)^{s+1} }C'^{s} \int_{v_*}^v \inf_{u\in[u_*-\epsilon,u_*+\epsilon]}\left|\frac{\partial r}{\partial v}\right|^{s} dv' \\
        &\to\infty\text{ as }v\to 0.\label{eqn:divergence}
\end{align}

It now follows from Lemma \ref{lemma:integralbalance} (with $\mathcal{D} = \iota_{\Q}((u_*-\epsilon,u_*+\epsilon)\times(v_*,0))$, $f=\partial_q\tilde{r}$, $h=\frac{\partial r}{\partial v}\frac{\partial v}{\partial q}$, and $\mathcal{D}_v = \iota_{\Q}((u_*-\epsilon,u_*+\epsilon)\times(v_*,v])$) and \eqref{EqDr}, \eqref{eqn:divergence} that for $v<0$ sufficiently small we have
    \begin{align}
 \int_{\iota_{\Q}([u_*-\epsilon,u_*+\epsilon]\times[v_*,v])}\left|\frac{\partial r}{\partial u}\frac{\partial u}{\partial q}\right|^{s}dpdq&\geq \frac{1}{2}\int_{\iota_{\Q}([u_*-\epsilon,u_*+\epsilon]\times[v_*,v])}\left|\frac{\partial r}{\partial v}\frac{\partial v}{\partial q}\right|^{s}dpdq\label{eqn:ineq}\\
    &\to\infty\text{ as }v\to0.\label{eqn:notL2}
\end{align}
Combining this with Lemma \ref{lemma:derivativecontrol}, we conclude that
\begin{align}
    \int_{\iota_{\Q}([u_*-\epsilon,u_*+\epsilon]\times[v_*,v])}\left|\frac{\partial r}{\partial u}\frac{\partial u}{\partial p}\right|^{s}dpdq&\geq \frac{1}{\alpha^{s}} \int_{\iota_{\Q}([u_*-\epsilon,u_*+\epsilon]\times[v_*,v])}\left|\frac{\partial r}{\partial u}\frac{\partial u}{\partial q}\right|^{s}dpdq\label{eqn:bound}\\
    &\to\infty\text{ as }v\to0\label{eqn:ineq1}
\end{align}
which shows \eqref{EqBlowUp}.

Combining \eqref{eqn:ineq}, \eqref{eqn:bound} and applying Lemma \ref{lemma:derivativecontrol}, we obtain
\begin{align}
         \int_{\iota_{\Q}([u_*-\epsilon,u_*+\epsilon]\times[v_*,v])}\left|\frac{\partial r}{\partial u}\frac{\partial u}{\partial p}\right|^{s}dpdq&\geq\frac{1}{2\alpha^{s}} \int_{\iota_{\Q}([u_*-\epsilon,u_*+\epsilon]\times[v_*,v])}\left|\frac{\partial r}{\partial v}\frac{\partial v}{\partial q}\right|^{s}dpdq\\
         &\geq\frac{1}{2\alpha^{2s}} \int_{\iota_{\Q}([u_*-\epsilon,u_*+\epsilon]\times[v_*,v])}\left|\frac{\partial r}{\partial v}\frac{\partial v}{\partial p}\right|^{s}dpdq.\label{eqn:ineq2}
\end{align}

Since we chose $\alpha<\frac{1}{2}$ initially, we conclude from \eqref{eqn:ineq2} that
\begin{equation}
\begin{split}
            \frac{|| \frac{\partial r}{\partial v}\frac{\partial v}{\partial p}\cdot \mathbbm{1}_{\iota_{\Q}([u_* - \epsilon, u_* + \epsilon]\times [v_*,v))}||_{L^{s}(\tilde{U})}}{||\frac{\partial r}{\partial u}\frac{\partial u}{\partial p} \cdot \mathbbm{1}_{\iota_{\Q}([u_* - \epsilon, u_* + \epsilon]\times [v_*,v))}||_{L^{s}(\tilde{U})}}&\leq 2^{1/s}\alpha^{2}<1
\end{split}
\end{equation}
  which shows \eqref{eqn:limitnot1} and contradicts Lemma \ref{lemma:integralbalance} as outlined before.
    \end{proof}

\subsection{The Corner Extension}\label{The Corner Extension}
The following example (inspired by \cite[Example 4.1]{CameronSbierski}) shows that Theorem \ref{thm:Christoffelblowup} is false if we drop the assumption that $\iota_\Q$ is locally $C^0$-equivalent to $\overline{Q}$. 

\begin{example}
    Let $s>1$ and choose $\beta\in\left(\frac{1}{2},1\right)$ such that $\beta>\frac{2s-1}{3s-1}$ and $s\beta>\text{max}\{1,s-1\}$.\footnote{Note that in \cite[Example 4.1]{CameronSbierski} we only required $\beta\in\left(\frac{1}{2},1\right)$.} 
   We will construct an example of a strongly spherically symmetric $C^0\cap W^{1,s}_{\loc}$-extension of $(\Mw,\gw)$ across $\{v=0\}$ satisfying all assumptions of Theorem \ref{thm:Christoffelblowup} except that $\iota_{\Q}$ will not be $C^0$-equivalent to $\overline{Q}$. 

Let $r(u,v)=1+|v|^{1-\beta}$ and suppose $\Omega(u,v)=1$ for all $(u,v)\in(-1,1)\times(-1,0)$. Note that $r$ and $\Omega$ extend as continuous strictly positive functions to $[-1,1]\times(-1,0]$ and $r(u,0)$ is independent of $u$. Let $\tilde{\Q}$ be the 1+1-dimensional Lorentzian manifold with global coordinates $(p,q)\in(-1,1)\times(-1,1)$ and let $\tilde{g}_{\beta,\tilde{\Q}}$ be the metric on $\tilde{\Q}$ given in these coordinates by 
    \begin{equation}
        \tilde{g}_{\beta,\tilde{\Q}}=\begin{cases}
            -\frac{1}{(1-\beta)}\left(dp\otimes dq + dq\otimes dp\right)+ \frac{2\beta p}{(1-\beta)^2q}dq\otimes dq & |p| < |q|^{\frac{\beta}{1-\beta}}, \text{ }-1<q<0;\\
            -\frac{1}{(1-\beta)}\left(dp\otimes dq + dq\otimes dp\right) +\frac{2\beta }{(1-\beta)^2}|q|^{\frac{2\beta-1}{1-\beta}}dq\otimes dq & \text{ }p\leq -|q|^{\frac{\beta}{1-\beta}}, -1<q<0;\\
            -\frac{1}{(1-\beta)}\left(dp\otimes dq + dq\otimes dp\right) - \frac{2\beta }{(1-\beta)^2}|q|^{\frac{2\beta-1}{1-\beta}}dq\otimes dq & \text{ }p\geq|q|^{\frac{\beta}{1-\beta}}, -1<q<0;\\
            -\frac{1}{(1-\beta)}\left(dp\otimes dq + dq\otimes dp\right)  & \text{ }0\leq q<1.            
        \end{cases}
    \end{equation}
    Note that the restriction $\beta\in\left(\frac{1}{2},1\right)$ ensures $\tilde{g}_{\beta}$ is continuous. 

  Define $\iota_{\beta,\Q}:\Q\hookrightarrow\tilde{\Q}$ in terms of coordinates $(u,v)\in(-1,1)\times(-1,0)$ by
\begin{equation}
    \begin{split}
       \iota_{\beta,\Q}(u,v)&=(u|v|^\beta,-|v|^{1-\beta}).
    \end{split}
\end{equation}
It is shown in \cite[Example 4.1]{CameronSbierski} that $\iota_{\beta,\Q}$ defines a $C^0$-extension of $(\Q,g_\Q)$ across $\{v=0\}$ into $(\tilde{\Q},\tilde{g}_{\beta,\tilde{\Q}})$.

   Consider the metric on $\tilde{M}:=\tilde{\Q}\times \Sp^2$ given in coordinates $(p,q,\theta^A)$ by (abusing notation slightly)
    \begin{equation}\label{eqn:fullcornermetric}
        \tilde{g}_{\beta}=\tilde{g}_{\beta,\tilde{\Q}}+\tilde{r}^2d\sigma^2_{\Sp^2}
    \end{equation}
    where $\tilde{r}(p,q)=1-q$.
    We see that
$\iota_{\beta,SSS}:=\iota_{\beta,\Q}\times{\id}_{\Sp^2}:\Mw\hookrightarrow \tilde{M}$ defines a strongly spherically symmetric $C^0$-extension of $(\Mw,\gw)$ across $\{v=0\}$ into $(\tilde{M},\tilde{g}_{\beta})$. Note that for any $u_*\in(-1,1)$ we have $\lim_{v\to0}\iota_{\beta,\Q}(u_*,v)=(0,0)$ in $(p,q)$ coordinates, so we are in case (a) of Theorem \ref{thm:boundaryrigidity}. In particular, $\iota_{\beta,\Q}$ is not locally $C^0$-equivalent to $\overline{Q}$. However, we have
\begin{equation}
    \begin{split}
       \int_{v_*}^v \inf_{u \in (-1,1)}\left|\partial_vr(u,v')\right|^{s}dv'&=(1-\beta)^{s}\int_{v_*}^v(-v')^{-s\beta}dv'\\
       &=\frac{(1-\beta)^{s}}{s\beta-1}\left(|v|^{1-s\beta}-|v_*|^{1-s\beta}\right)\\
       &\to\infty\text{ as }v\to0
    \end{split}
\end{equation}
as assumed in Theorem \ref{thm:Christoffelblowup}.
It remains to check that $\tilde{g}_{\beta}\in W^{1,s}_{\loc}$. To do this, it suffices to check that the function $f:(-1,1)\times(-1,1)\to\R$ defined by 
\begin{equation}
    f(p,q)=\begin{cases}
        \frac{p}{q} & |p| < |q|^{\frac{\beta}{1-\beta}}, \text{ }-1<q<0;\\
         |q|^{\frac{2\beta-1}{1-\beta}} & \text{ }p\leq -|q|^{\frac{\beta}{1-\beta}}, -1<q<0;\\
            - |q|^{\frac{2\beta-1}{1-\beta}}& \text{ }p\geq|q|^{\frac{\beta}{1-\beta}}, -1<q<0;\\
            0 & \text{ }0\leq q<1. 
    \end{cases}
\end{equation}
lies in $W^{1,s}(\tilde{\Q})$. We have
\begin{equation}
\begin{split}
    \int_{-1}^1\int_{-1}^1|\partial_pf(p,q)|^sdpdq&=\int_{-1}^0\int_{-|q|^{\frac{\beta}{1-\beta}}}^{|q|^{\frac{\beta}{1-\beta}}}|q|^{-s}dpdq\\
    &=2\int_{-1}^0|q|^{\frac{\beta}{1-\beta}-s}dq\\
    &<\infty
\end{split}
\end{equation}
since $s\beta>s-1$. Furthermore,
\begin{equation}
\begin{split}
    \int_{-1}^1\int_{-1}^1|\partial_qf(p,q)|^sdpdq&\leq\int_{-1}^0\int_{-|q|^{\frac{\beta}{1-\beta}}}^{|q|^{\frac{\beta}{1-\beta}}}\frac{|p|^s}{|q|^{2s}}dpdq+2\left(\frac{2\beta-1}{1-\beta}\right)^s\int_{-1}^0\int_{-1}^{-|q|^{\frac{\beta}{1-\beta}}}|q|^{\frac{s(3\beta-2)}{1-\beta}}dpdq\\
    &=\frac{2}{s+1}\int_{-1}^0|q|^{\frac{s(3\beta-2)+\beta}{1-\beta}}dq+2\left(\frac{2\beta-1}{1-\beta}\right)^s\int_{-1}^{0}\left(|q|^{\frac{s(3\beta-2)}{1-\beta}}-|q|^{\frac{s(3\beta-2)+\beta}{1-\beta}}\right)dq\\
    &<\infty
\end{split}
\end{equation}
since $\beta>\frac{2s-1}{3s-1}\implies\frac{s(3\beta-2)}{1-\beta}>-1$.

Since $f$ is continuous, we conclude that $f\in W^{1,s}(\tilde{\Q})$ and hence $\tilde{g}_{\beta}\in W^{1,s}_{\loc}$.


\end{example}

\section{Application to  weak null singularities in the interior of spherically symmetric black holes}\label{Application to  weak null singularities in the interior of spherically symmetric black holes}
In this section we apply Theorem \ref{thm:Christoffelblowup} to the weakly singular Cauchy horizon (the weak null singularity) arising in the interior of certain dynamical spherically symmetric black hole spacetimes. Recall that this theorem only rules out extensions where $\iota_\Q$ is locally $C^0$-equivalent to $\overline{\Q}$. Theorem \ref{thm:boundaryrigidity} implies that if this is not the case then $\lim_{v\to0}\iota_{\Q}(u,v)$ must exist for all $u\in(-1,1)$ and be independent of $u$. In the examples that follow we will rule out this scenario. In Section \ref{RNV} we do this for the Reissner-Nordstr\"{o}m-Vaidya spacetime using the global volume argument of Proposition \ref{prop:cornervolumecondition} (see also \cite[Section 4]{CameronSbierski}). In Section \ref{DLO} we consider spacetimes arising from small and generic spherically symmetric perturbations to Reissner-Nordstr\"{o}m initial data in the Einstein-Maxwell-scalar field system. In this case the corner extension can be ruled out even locally since $\partial_ur$ extends continuously to the Cauchy horizon with $\partial_ur<0$.
\subsection{Reissner-Nordstr\"{o}m-Vaidya spacetimes}\label{RNV}
The Reissner-Nordstr\"{o}m-Vaidya (RNV) spacetime $(M_{RNV},g_{RNV})$ is given by $M_{RNV}=(v_0,\infty)\times(0,\infty)\times\Sp^2$ with canonical $(v,r)$-coordinates on the first two factors and
\begin{equation}
\begin{split}
        g_{RNV}&=-\left(1-\frac{2\varpi(v)}{r}+\frac{e^2}{r^2}\right)dv^2+dv\otimes dr+dr\otimes dv+r^2\gamma_{AB}d\theta^A\otimes d\theta^B
\end{split}
\end{equation}
where $\varpi(v)=\varpi(\infty)-\beta v^{-p}$, $\beta>0$, $p>1$, $e>0$ and $v_0>0$ is sufficiently large that $\varpi(v_0)>e$. This models the continuous influx of null dust into a subextremal Reissner-Nordstr\"{o}m black hole decaying with a tail $\rho\sim v^{-(p+1)}$.
As in \cite[Section 4.2]{Sbie22a}, we define a second null coordinate, $u\in(u_{\mathcal{H}^+},u_T)$ for some $-\infty<u_{\mathcal{H}^+}<u_T<\infty$, using the method of characteristics.\footnote{Indeed, by \cite[Proposition 4.29]{Sbie22a} we have $u_T=0$.} The metric in double null coordinates $(u,v, \theta^A)$ takes the form\footnote{Note the factor of $2$ difference between the definition of $\Omega^2$ here and that in \cite{Sbie22a}.}
\begin{equation}
    g=-\Omega^2(u,v)\left(du\otimes dv+dv\otimes du\right)+r^2(u,v)\gamma_{AB}d\theta^A\otimes d\theta^B 
\end{equation}
with $\Omega^2 = - \rd_u r$, see Step 3 of the proof of \cite[Proposition 4.25]{Sbie22a}. 
We also define the rescaled null coordinate $V(v):=-e^{\kappa_-(\infty)v}$, where $\kappa_-(\infty)<0$ is the surface gravity of the Cauchy horizon as defined in \cite{Sbie22a}. The $(u,V,\theta^A)$ coordinates are defined on $\Q_{RNV} \times \Sp^2$, with $\Q_{RNV} = (u_{\mathcal{H}^+}, u_T) \times (V_0, 0)$, where  $V_0 = V(v_0)$. Note that this corresponds to a subset of $M_{RNV}$, see \cite[Figure 10]{Sbie22a}. In these coordinates, the metric
\begin{equation}
    g=-\underbrace{\frac{\Omega^2}{\kappa_-(\infty)V}}_{=: \overline{\Omega}^2}\left(du\otimes dV+dV\otimes du\right)+r^2\gamma_{AB}d\theta^A\otimes d\theta^B
\end{equation}
extends continuously to $\overline{\Q_{RNV}} \times \Sp^2$, where $\overline{\Q_{RNV}} = (u_{\mathcal{H}^+}, u_T) \times (V_0, 0]$ by \cite[Proposition 4.25]{Sbie22a}. The weakly singular Cauchy horizon is at $\{V=0\}$. 
 
\begin{cor}[Reissner-Nordstr\"{o}m-Vaidya]\label{cor:RNV}
    For $s>1$, there is no strongly spherically symmetric $C^0 \cap W^{1,s}_{\loc}$-extension of the Reissner-Nordstr\"{o}m-Vaidya spacetime across the weakly singular Cauchy horizon.
\end{cor}
As in \cite{Sbie22a}, we remark that our inextendibility result does not require mass inflation \cite{Ori:1991zz,Poisson:1989zz,Poisson:1990eh} as an assumption. 
\begin{proof} 
Suppose there exists a strongly spherically symmetric $C^0 \cap W^{1,s}_{\loc}$-extension $\iota:=\iota_{\Q_{RNV}}\times\id_{\Sp^2} : \Q_{RNV} \times \Sp^2 \to \tilde{\Q} \times \Sp^2$  across $\{V=0\}$.\footnote{We will lead this statement to a contradiction. Note that as stated, Corollary \ref{cor:RNV} applies to the global spacetime $M_{RNV}$. However, it is sufficient to prove this local result since, as discussed below \cite[Figure 10]{Sbie22a}, one can always increase $v_1 >v_0$ to construct null coordinates that cover an arbitrarily large part of the weakly singular Cauchy horizon.} 

By Remark \ref{RemEquivDefAcross} there exists $u_* \in (u_{\mathcal{H}^+}, u_T)$ such that $\lim_{V \to 0} \iota_{\Q_{RNV}}(u_*,V) \in \tilde{\Q}$ exists. We now verify the remaining assumptions of Proposition \ref{prop:cornervolumecondition}.

By Claim 6 of \cite[Section 4.2]{Sbie22a} we can choose $u_1\in(u_{\mathcal{H}^+},u_T)$ such that $\lim_{v\to\infty}r(u_1,v)=r_-(\infty)>0$ and by Claim 4 of \cite[Section 4.2]{Sbie22a} we have $\lim_{v\to\infty}r(u_{\mathcal{H}^+},v)=r_+(\infty)>r_-(\infty)$.
Hence for any $v_*>v_0$ we have
\begin{equation}
    \begin{split}
        \vol_{g_{\Q_{RNV}}}\left((u_{\mathcal{H}^+},u_1)\times(V(v_*),0)\right):=
        &\int_{v_*}^\infty\int_{u_{\mathcal{H}^+}}^{u_1}\Omega(u',v')^2du'dv'\\
        =&-\int_{v_*}^\infty\int_{u_{\mathcal{H}^+}}^{u_1}\partial_ur(u',v') du'dv'\\
        =&\int_{v_*}^\infty \left(r(u_{\mathcal{H}^+},v') -r(u_1,v'
        )\right)dv'\\
        =&\infty
    \end{split}
\end{equation}
because the integrand converges to $r_+(\infty)-r_-(\infty)>0$ as $v'\to\infty$. It follows from Proposition \ref{prop:cornervolumecondition} that $\iota_{\Q_{RNV}}$ is locally $C^0$-equivalent to $\overline{\Q_{RNV}}$ at $(u_*,0) \in \overline{\Q_{RNV}}$.


Next we show that the blow-up condition \eqref{eqn:blowupassumption} holds. Let $V_*<V<0$ and $\epsilon>0$ be sufficiently small that $[u_*-\epsilon,u_*+\epsilon]\subset (u_{\mathcal{H}^+},u_T)$. By the proof of \cite[Proposition 4.25]{Sbie22a}, there exists a constant $C>0$ such that\footnote{This estimate is stated in \cite{Sbie22a} for a fixed value of $u\in(u_{\mathcal{H}^+},u_T)$ however it is straightforward to check that it holds uniformly for all $u\in[u_*-\epsilon,u_*+\epsilon]$.}
\begin{equation} \begin{split}
          \partial_Vr(u,V)&\leq -Ce^{-\kappa_-(\infty)v(V)}v(V)^{-(p+1)}\\
     &=\frac{C}{V}\left|\frac{\kappa_-(\infty)}{\log|V|}\right|^{p+1}\\
     \implies      \int_{_{V_*}}^V\inf_{u\in[u_*-\epsilon,u_*+\epsilon]}\left|\partial_Vr(u,V')\right|^{s}dV'&\geq C^{s}|\kappa_-(\infty)|^{s(p+1)}\int_{V_*}^V|V'|^{-s}\left|\log|V'|\right|^{-s(p+1)}dV'\\
     &\to\infty\text{ as }V\to0
\end{split}
\end{equation}  
for $s>1$.
 Hence, Theorem \ref{thm:Christoffelblowup} leads the statement we started with to a contradiction.
\end{proof}
\subsection{Dafermos-Luk-Oh spacetimes}\label{DLO}
We now consider spacetimes $(M_{DLO} = \mathfrak{Q}_{DLO} \times \Sp^2,g_{DLO})$ arising from small and generic spherically symmetric perturbations of asymptotically flat two-ended subextremal Reissner-Nordstr\"{o}m initial data for the Einstein-Maxwell-scalar field system. Following \cite{Sbie22a}, we refer to this class of spacetimes
\begin{wrapfigure}[11]{r}{0.46\textwidth}
 \centering
  \def\svgwidth{7.5cm}
   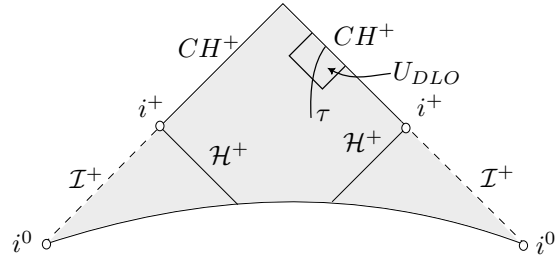 
   \caption{Penrose-style diagram of the DLO spacetimes.}
   \label{FigDLO}
\end{wrapfigure}
 as \textit{Dafermos-Luk-Oh} spacetimes. Their Penrose diagram is given in Figure \ref{FigDLO}. The black hole interior is bounded to the future by a bifurcate weakly singular Cauchy horizon $CH^+$: the weak null singularity. The metric extends continuously to $\overline{\mathfrak{Q}_{DLO}} \times \Sp^2 = \mathfrak{Q}_{DLO} \times \Sp^2 \cup CH^+$.

\begin{cor}[Dafermos-Luk-Oh spacetimes]\label{cor:DLO}
Suppose that $(M_{DLO} = \mathfrak{Q}_{DLO} \times \Sp^2,g_{DLO})$ arises as a solution to the Einstein-Maxwell-scalar field system in spherical symmetry from generic and sufficiently small perturbations of asymptotically flat two-ended subextremal Reissner-Nordstr\"{o}m initial data as  considered in \cite{LukOh19I}, \cite{LukOh19II}, \cite{Gautam}.\footnote{It should be mentioned that in \cite{LukOh19I} Luk and Oh do not only consider globally small perturbations of subextremal RN initial data that lead to what we call here a DLO spacetime, but they also treat ‘admissible’ large deviations from exact RN which possibly lead to the closing off of the Cauchy horizon, transitioning into a spacelike singularity in the interior. We do not discuss
the latter case here.}  Then there is no strongly spherically symmetric $C^0\cap W^{1,s}_{\loc}$-extension of $(M_{DLO},g_{DLO})$ across the Cauchy horizon $CH^+$. 
\end{cor}

\begin{proof}[Sketch of proof]
    Assume that $\iota := \iota_{\mathfrak{Q}_{DLO}} \times \id_{\Sp^2} : \mathfrak{Q}_{DLO} \times \Sp^2 \to \tilde{Q} \times \Sp^2$ is a strongly spherically symmetric $C^0\cap W^{1,s}_{\loc}$-extension of $(M_{DLO},g_{DLO})$ across the Cauchy horizon $CH^+$. This means there exists a $C^1$ future-directed causal curve $\tau : [-1,0) \to \mathfrak{Q}_{DLO}$ which has a limit point on the Cauchy horizon (in the quotient spacetime $\overline{\mathfrak{Q}_{DLO}}$) and such that $\iota_{\mathfrak{Q}_{DLO}} \circ \tau$ also has a limit point in $\tilde{\Q}$. If $\tau$ has a limit point on the bifurcation sphere, then the same argument as in \cite[Step 2.2 of proof of Theorem 4.1]{Sbie22a} shows that there also exists a radial null geodesic in $M_{DLO}$ which has a limit on the Cauchy horizon away from the bifurcation sphere as well as in the assumed $C^0 \cap W^{1,s}_{\loc}$-extension. Since the geometry of the left and the right Cauchy horizon are comparable, we can thus assume without loss of generality that $\tau$ approaches the right Cauchy horizon away from the bifurcation sphere.

We now consider a portion of the global spacetime near the limit point of $\tau$ on the right Cauchy horizon given by $U_{DLO}=\Q_{DLO}\times \Sp^2:=(u_0,u_1)\times(V_1,1)\times \Sp^2$, where $(u,V)$ are coordinates on the first two factors and $-\infty < u_0 < u_1 < \infty$, see Figure \ref{FigDLO}. The metric in these coordinates is
\begin{equation}
    g_{DLO}=\underbrace{-\Omega^2\left(du\otimes dV+dV\otimes du\right)}_{g_{DLO,\Q}}+r^2\gamma_{AB}d\theta^A\otimes d\theta^B
\end{equation}
where $\Omega$ and $r$ are smooth strictly positive functions on $\Q_{DLO}$ which extend continuously to strictly positive functions on $\overline{\Q}_{DLO}:=(u_0,u_1)\times(V_1,1]$, see \cite[Theorem 5.5]{LukOh19I}. The hypersurface $\{V=1\}$ corresponds to the portion of the (right) Cauchy horizon intersected with $U_{DLO}$. Clearly, $\iota$ restricts to a strongly spherically symmetric $C^0\cap W^{1,s}_{\loc}$-extension $\iota_{\Q_{DLO}} \times \id_{\Sp^2}$ of $(U_{DLO},g_{DLO})$ across  $\{V=1\}$. By Remark \ref{RemEquivDefAcross} there exists $u_* \in (u_0,u_1)$ such that $\lim_{V \to 1}\iota_{\Q_{DLO}}(u_*,V)$ exists in $\overline{\Q_{DLO}}$.

By \cite{Gautam} in combination with either \cite{Daf05a} or \cite{LukOhShla23} it follows that $\rd_u r$ extends continuously to $\{V=1\}$ as a strictly  \emph{negative} function -- and thus $r(u,1)$ is in particular not constant. Assume now that $\iota_{\Q_{DLO}}$ were not locally $C^0$-equivalent to $\overline{\Q_{DLO}}$ at $(u_*,1)$. Then by Theorem \ref{thm:boundaryrigidity}, $\iota_{\Q_{DLO}}(u,1) = \tilde{p}_{u_*}$ must be independent of $u$ -- and we also have, by assumption, that $r \circ \iota_{\Q_{DLO}^{-1}}$ extends continuously to $\tilde{p}_{u_*}$. This, however, is in contradiction to $r(u,1)$
 not being constant. Hence, $\iota_{Q_{DLO}}$ must be locally $C^0$-equivalent to $\overline{\Q_{DLO}}$ at $(u_*,1)$.

 Let $\epsilon>0$ be sufficiently small that $[u_*-\epsilon,u_*+\epsilon]\subset(u_0,u_1)$. From the proof of \cite[Proposition C.1]{LukOh19I}, there exists a constant $C>0$ such that for $(u,V)\in [u_*-\epsilon,u_*+\epsilon]\times [V_1,1)$ we have
    \begin{equation}
        C^{-1}\leq r\leq C,\quad \left|\partial_u(r\partial_Vr)\right|\leq C.
    \end{equation}
It follows that 
\begin{equation}
    \begin{split}
        \left|r(u,V)\partial_Vr(u,V)-r(u_*,V)\partial_Vr(u_*,V)\right|&\leq C|u-u_*|\\
        &\leq C\epsilon\\
        \implies\inf_{u\in[u_*-\epsilon,u_*+\epsilon]}\left|\partial_Vr(u,V)\right|&\geq \frac{1}{C}\left|\partial_Vr(u_*,V)\right|-\epsilon.
    \end{split}
\end{equation}
By \cite[Proposition C.1]{LukOh19I} we have $\lim_{V\to1}||\partial_Vr(u_*,V)\cdot\mathbbm{1}_{(V_1,V)}||_{L^s(V_1,1)}=\infty$, so it follows that 
\begin{equation}
    \int_{V_1}^V \inf_{u\in[u_*-\epsilon,u_*+\epsilon]}  \left|\partial_Vr(u,V')\right|^{s} dV'\to\infty\text{ as }V\to1\;.
\end{equation}
    
    Hence by Theorem \ref{thm:Christoffelblowup}, we obtain the contradiction that $\iota_{\Q_{DLO}} \times \id_{\Sp^2}$ is not a strongly spherically symmetric $C^0\cap W^{1,s}_{\loc}$-extension.
    

\end{proof}

\bibliographystyle{amsplain}
\bibliography{bibliography.bib}
\end{document}